\documentclass[acmsmall]{acmart}
\AtBeginDocument{%
  }

\usepackage{bm}
\usepackage[framemethod=1]{mdframed}

\usepackage{algorithm}
\usepackage{algorithmic}

\newtheorem{theorem}{Theorem}

\newtheorem{proposition}[theorem]{Proposition}

\newtheorem{definition}{Definition}

\newtheorem*{corollary2_1*}{Corollary 2.1}
\newtheorem*{corollary5_1*}{Corollary 5.1}

\newcounter{example}

\newenvironment{proofsketch}[1][A Proof Sketch]{\begin{trivlist}
\item[\hskip \labelsep {\bfseries #1}]}{\end{trivlist}}

\usepackage{float}
\newfloat{algorithm}{t}{lop}

\usepackage{tabularx}
\usepackage{bm}
\allowdisplaybreaks

\newsavebox{\ieeealgbox}

\usepackage[caption=false]{subfig}

\setcopyright{acmlicensed}
\copyrightyear{2026}
\acmYear{2026}
\acmDOI{XXXXXXX.XXXXXXX}

\begin{document}

\title{Optimal Parallel Scheduling under Concave Speedup Functions}

\author{Chengzhang Li}
\affiliation{%
  \institution{Department of Electrical and Computer Engineering, The Ohio State University}
  \city{Columbus}
  \state{OH}
  \country{USA}
  }

\author{Peizhong Ju}
\affiliation{%
  \institution{Department of Computer Science, University of Kentucky}
  \city{Lexington}
  \state{KY}
  \country{USA}
  }

\author{Atilla Eryilmaz}
\affiliation{%
  \institution{Department of Electrical and Computer Engineering, The Ohio State University}
  \city{Columbus}
  \state{OH}
  \country{USA}
  }

\author{Ness B. Shroff}
\affiliation{%
  \institution{Department of Electrical and Computer Engineering, The Ohio State University}
  \city{Columbus}
  \state{OH}
  \country{USA}
  }

\renewcommand{\shortauthors}{Anonymous authors.}

\begin{abstract}
Efficient scheduling of parallel computation resources across multiple jobs is a fundamental problem in modern cloud/edge computing systems for many AI-based applications. Allocating more resources to a job accelerates its completion, but with diminishing returns. Prior work (heSRPT) solved this problem only for some specific speedup functions with an exponential form, providing a closed-form solution. However, the general case with arbitrary concave speedup functions---which more accurately capture real-world workloads---has remained open.  

In this paper, we solve this open problem by developing optimal scheduling algorithms for parallel jobs under general concave speedup functions. 
We first discover a fundamental and broadly-applicable rule for optimal parallel scheduling, namely the \emph{Consistent Derivative Ratio (CDR) Rule}, which states that the ratio of the derivatives of the speedup functions across active jobs remains constant over time. 
To efficiently compute the optimal allocations that satisfy the CDR Rule, we propose the \emph{General Water-Filling (GWF)} method, a more general version of classical water-filling in wireless communications.
Combining these insights, we design the \emph{SmartFill} Algorithm to solve the general scheduling problem. Unlike heSRPT, which always allocates resources to all active jobs, SmartFill selectively determines which jobs should receive resources and how much they should be allocated. For a broad class of so-called \emph{regular} speedup functions, SmartFill yields closed-form optimal solutions, while for non-regular functions it efficiently computes the optimum with low complexity. Numerical evaluations show that SmartFill can substantially outperform heSRPT  across a wide range of concave speedup functions. 
\end{abstract}

\maketitle

\section{Introduction}
\label{sec:introduction}

Parallel computing has become indispensable in modern computing systems~\cite{navarro2014survey}.  
At the cloud scale, data centers such as Microsoft Azure~\cite{collier2015microsoft}, Amazon Web Services (AWS)~\cite{mathew2014overview}, and Google Cloud~\cite{bisong2019overview} operate with tens of thousands of GPUs per site.  
Even at the edge, powerful devices such as the NVIDIA A100 GPU provide up to 6,912 CUDA cores that can run in parallel~\cite{choquette2021nvidia}.  
Modern AI workloads are explicitly designed to exploit this massive parallelism. For example, training deep neural networks such as CNNs~\cite{li2021survey}, Transformers~\cite{vaswani2017attention}, and large language models~\cite{naveed2025comprehensive} leverages parallel GPUs to significantly accelerate training~\cite{strigl2010performance,hu2022survey}, while GPU-based inference of AI models reduces latency in many applications~\cite{wang2019unified,sheng2023flexgen}.

When multiple jobs compete for parallel-computing resources, scheduling becomes a critical task~\cite{ye2024deep}.  
A common objective is to minimize job completion times or related performance metrics~\cite{avrahami2003minimizing}.  
The efficiency of any scheduling policy depends heavily on the \emph{speedup function}, which describes how a job’s processing speed increases as more parallel resources are allocated.  
If the speedup function is linear, the classic SRPT algorithm~\cite{schrage1966queue} is optimal.  
In practice, however, speedup functions are typically concave, reflecting diminishing returns as additional resources are assigned~\cite{ghanbarian2024optimal,harchol2021open}.  
This makes the scheduling problem substantially more challenging.

Berg \emph{et al.}~\cite{berg2020hesrpt} studied this problem under a special family of speedup functions $s(\theta)=\theta^p$ for $0<p<1$.  
They derived elegant theoretical results and a closed-form optimal scheduling policy in this setting.  
However, this family fails to capture the general behavior of concave functions.  
A fundamental limitation is that $s'(0)=\infty$, implying that allocating even an infinitesimal amount of resource to a job yields infinite benefit.  
Consequently, under the optimal schedule in this setting, every job must receive a positive allocation until completion.  
This property is unrealistic in real-world systems, where allocating a small amount of resource provides limited benefit.  
Harchol-Balter~\cite{harchol2021open} identified the problem of scheduling under general concave speedup functions as ``an entirely open problem.''

In this paper, we solve this open problem by developing optimal scheduling algorithms for parallel jobs with \emph{general concave speedup functions}.  
Our system model and objective are identical to those in~\cite{berg2020hesrpt}, except that we allow arbitrary concave speedup functions instead of restricting to $s(\theta)=\theta^p$.  
This generalization substantially broadens applicability to diverse workloads and computing platforms.

We make two key theoretical contributions.  
First, we establish the \emph{Consistent Derivative Ratio (CDR) Rule}, which states that under any optimal policy, the ratio of the derivatives of the speedup functions across active jobs remains constant over time.  
This property dramatically reduces the search space, since any optimal schedule must satisfy it.  

Second, we propose the \emph{General Water-Filling (GWF)} algorithm, which efficiently computes allocations that satisfy the CDR Rule.  
GWF generalizes the classical water-filling algorithm in wireless communications~\cite{tse2005fundamentals}: while the classical setting assumes identical channel widths, GWF accommodates ``bottles'' of different widths and shapes.  
For a broad family of \emph{regular speedup functions}, including $s(\theta)=\theta^p$, GWF admits closed-form solutions.  
For non-regular concave functions, it produces efficient numerical solutions.

Building on these results, we design the \emph{SmartFill} algorithm, which integrates the CDR Rule and GWF into a complete solution.  
Unlike heSRPT, which allocates some resource to every active job, SmartFill intelligently determines which jobs should receive resources and how much should be allocated.  
For regular speedup functions, SmartFill yields closed-form optimal schedules; for general concave functions, it provides efficient numerical solutions.  
Through numerical evaluation, we show that SmartFill consistently outperforms heSRPT  across a wide range of speedup functions.  

Finally, we discuss the broader applicability of the results in this paper. 
The CDR Rule is in fact a very general principle in parallel scheduling.  
It continues to hold under more general settings, including scenarios with heterogeneous speedup functions across jobs and time-varying total system resources, thereby significantly reducing the search space of possible schedules.  
However, the GWF and SmartFill algorithms do not directly extend to these settings, as it remains unclear how to determine the job completion order in such general cases.  
We identify this as an important open problem for future research.

\section{System Model and Problem Statement}\label{sec:model}
We consider the \emph{exact same} system model as in \cite{berg2020hesrpt}, with the only difference being the assumption on the speedup function. 
Specifically, we assume the speedup function is a general concave function, rather than an exponential one. 
While some notations used here differ slightly from those in \cite{berg2020hesrpt}, they are mathematically equivalent.

We consider a system with $M$ parallelizable jobs, each available at time $t=0$. Let $x_i$ denote the size of job $i$ for $i = 1, 2, \ldots, M$.
Without loss of generality, we assume that job sizes follow a non-increasing order:
\begin{equation*}
    x_1 \geq x_2 \geq \cdots \geq x_M.
\end{equation*}

The system includes a divisible server that allocates computation resource among the jobs at any time $t > 0$. Let $B$ denote the total  bandwidth of the server, and let $\theta_i(t)$ denote the amount of resource allocated to job $i$ at time $t$. The allocations must satisfy the constraint:
\begin{equation}\label{eq:cst1}
    \sum_{i=1}^M \theta_i(t) \leq B, \quad \forall t > 0.
\end{equation}
We assume that $\theta_i(t)$ is \emph{right-continuous} w.r.t. $t$. This assumption rules out discontinuities at isolated points\footnote{Alternatively, assuming left-continuity yields the same theoretical results throughout the paper.}.

We assume all jobs follow the same speedup function $s(\theta)$ defined in $\theta\in[0,B]$, which denotes the service rate for a job when allocated with bandwidth $\theta$. 
Let $s'(\theta)$ denote the derivative of $s(\theta)$.
We assume that the speedup function satisfies:
\begin{itemize}
    \item $s(0) = 0$,
    \item $s(\theta)$ is strictly increasing w.r.t. $\theta$,
    \item $s(\theta)$ is continuous and differentiable w.r.t. $\theta$,
    \item $s(\theta)$ is strictly concave w.r.t. $\theta$,
    \item $s'(\theta)$ is continuous  w.r.t. $\theta$.
\end{itemize}
Since $s(\theta)$ is strictly increasing and  strictly concave, we have $s'(\theta)>0$ and $s'(\theta)$  is strictly decreasing for $\theta\in[0,B]$.

Note that in \cite{berg2020hesrpt}, the authors considered a special family of speedup functions, $s(\theta)=\theta^p$, for $0<p<1$.
While this family is simple to analyze, it fails to capture the behavior of most concave functions. 
A key limitation is that these functions all satisfy $s'(0)=\infty$. This property is unrealistic in many real-world applications, where allocating a small amount of resource to a job provides only limited benefit rather than an infinite one. Moreover, this property fundamentally alters the structure of the optimal scheduler, as we will further discuss in Section \ref{sec:GWF}.

Let $Q_i(t_1,t_2)$ denote the amount completed service for job $i$ between time period $[t_1,t_2]$.
We have 
\begin{equation}\label{eq:service_t1_t2}
    Q_i(t_1,t_2)=\int_{t_1}^{t_2} s(\theta_i(t))\, dt, \quad \text{for } i = 1, 2, \cdots, M \text{ and } \forall t_1<t_2.
\end{equation}

Let $T_i$ denote the completion time of job $i$. 
For each job $i = 1, 2, \cdots, M$, we have:
\begin{equation}\label{eq:cst2}
    \theta_i(t) = 0, \quad \text{for } t > T_i.
\end{equation}
and 
\begin{equation}\label{eq:cst3}
    Q_i(0,T_i) = x_i, \quad \text{for } i = 1, 2, \cdots, M.
\end{equation}

Let $J$ denote the system performance metric, which is a weighted sum of job completion times:
\begin{equation}\label{eq1:obj}
    J = \sum_{i=1}^M w_i T_i.
\end{equation}
Here $w_i$ is the weight for job $i$.
Following the setting in \cite{berg2020hesrpt}, we assume that the weights follow a non-decreasing order:
\begin{equation*}
    w_1 \leq w_2 \leq \cdots \leq w_M.
\end{equation*}
This ordering implies that shorter jobs are given higher priority through larger weights. The choice of ${w_i}$ depends on the performance metric of interest. For example, when $w_1=w_2=\cdots=w_M=1$, $J$ corresponds to the mean flow completion time \cite{avrahami2003minimizing}. Alternatively, when $w_i=\frac{1}{x_i}$, $J$ represents the mean slowdown of the system \cite{harchol2002asymptotic}.

The objective of this paper is to find the optimal allocation $\theta_i(t)$ for each job $i$ and time $t$, so as to minimize $J$, given  $B$, $x_i$, $s(\theta)$, and $w_i$'s. Formally, we aim to solve the following optimization problem:
\begin{equation*}
\begin{aligned}
\text{OPT:} \quad \min_{\theta_i(t)} \quad & J  = \sum_{i=1}^M w_i T_i, \\
\text{subject to} \quad & \text{Constraints } \eqref{eq:cst1}, \eqref{eq:service_t1_t2},\eqref{eq:cst2}, \eqref{eq:cst3}.
\end{aligned}
\end{equation*}

The remainder of this paper is devoted to solving OPT. In Section~\ref{sec:CDR}, we present our first key result: a structural property of the optimal scheduling policy, the \emph{Consistent Derivative Ratio (CDR) Rule}, which significantly reduces the search space of OPT. Section~\ref{sec:GWF} introduces our second key result: the \emph{General Water-Filling (GWF)} algorithm, which efficiently computes allocations based on the CDR Rule. Finally, in Section~\ref{sec:SmartFill}, we combine these key results to develop the \emph{SmartFill} algorithm, which provides a complete solution to OPT.

\section{Consistent Derivative Ratio Rule}\label{sec:CDR}
In this section, we present the main result of this paper---the \emph{Constant Derivative Ratio (CDR) Rule}.

Let $\theta_i^*(t)$ denote an optimal scheduling policy to OPT.
The following theorem establishes that, under the optimal solution, the ratio of the derivatives of the speedup functions between any two jobs remains constant over time, as long as both jobs have a positive service rate.

\begin{theorem}\label{theorem:CDR}
    (\textbf{CDR Rule}): For any pair of jobs $i$ and $j$, there exists a constant $c_{i,j}$ such that
    \begin{equation}
        \frac{s'(\theta_i^*(t))}{s'(\theta_j^*(t))} = c_{i,j},
    \end{equation}
    for all $t$ such that $\theta_i^*(t) > 0$ and $\theta_j^*(t) > 0$.
\end{theorem}

\begin{figure}
  \centering
  \includegraphics[width=0.95\textwidth]{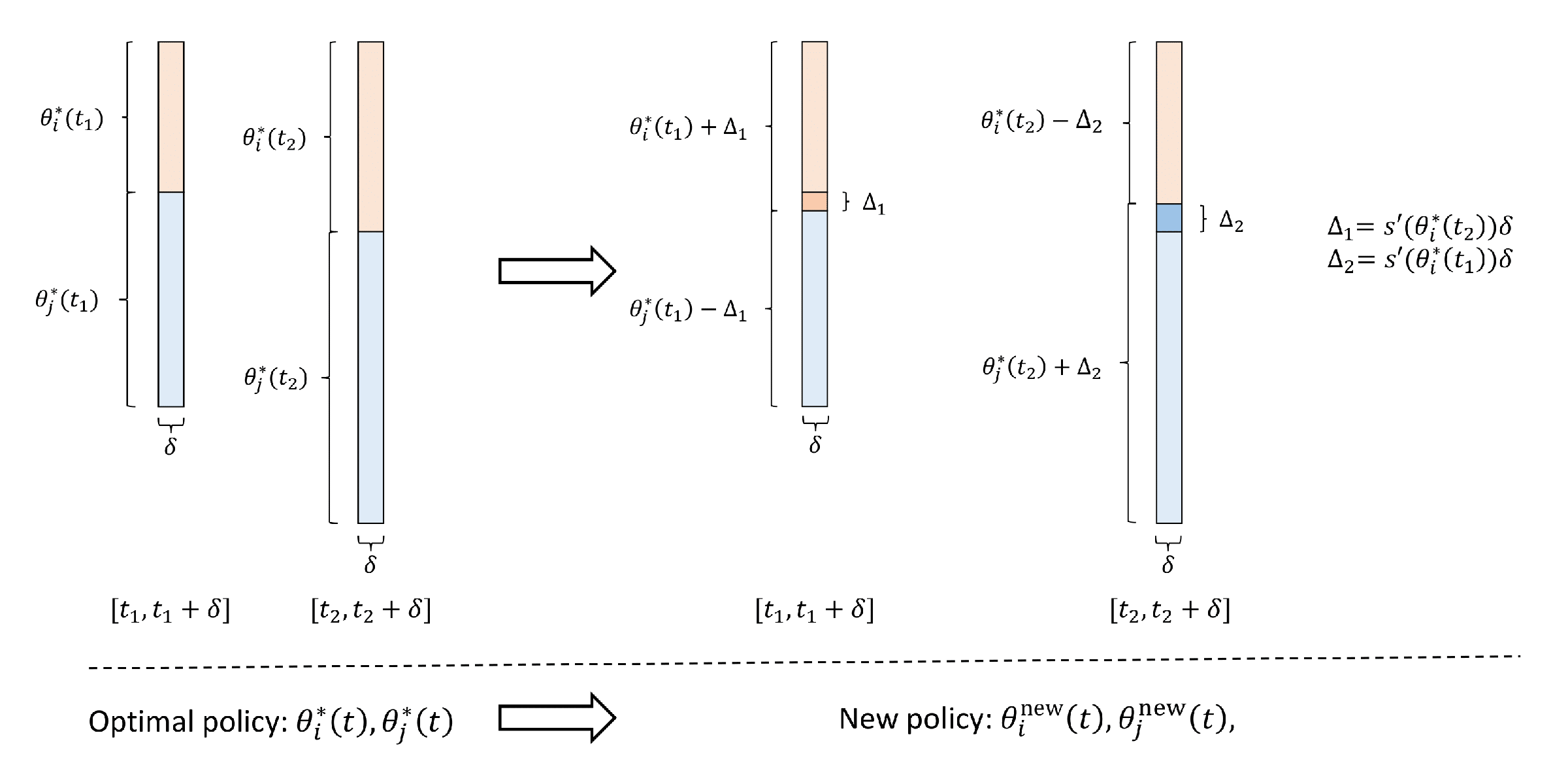}
  \caption{Idea of the proof of Theorem \ref{theorem:CDR}.}
  \label{fig:proof1} 
\end{figure}

\begin{proof}
We prove Theorem~\ref{theorem:CDR} by contradiction.

Assume the theorem is false. Then, there exist two time instances $t_1$ and $t_2$ such that
\begin{equation*}
    \frac{s'(\theta_i^*(t_1))}{s'(\theta_j^*(t_1))} > \frac{s'(\theta_i^*(t_2))}{s'(\theta_j^*(t_2))},
\end{equation*}
with $\theta_i^*(t_1) > 0$, $\theta_j^*(t_1) > 0$, $\theta_i^*(t_2) > 0$, and $\theta_j^*(t_2) > 0$.

Define
\begin{equation*}
    h_{i,j} = s'(\theta_i^*(t_1))s'(\theta_j^*(t_2)) - s'(\theta_i^*(t_2))s'(\theta_j^*(t_1)).
\end{equation*}
By the assumption, we have $h_{i,j} > 0$.

The idea of the proof is shown in Fig.~\ref{fig:proof1}.
We will construct a new scheduling solution $\theta_i^\text{new}(t)$ and $\theta_j^\text{new}(t)$ that decreases the objective value, while keeping all other scheduling decisions unchanged.


We construct $\theta_i^\text{new}(t)$ and $\theta_j^\text{new}(t)$ as follows. In the interval $[t_1, t_1 + \delta]$, we increase the service rate for job $i$ by a small amount $\Delta_1 = s'(\theta_i^*(t_2))\delta$ and decrease the service rate for job $j$ by the same amount:
\begin{align*}
    \theta_i^\text{new}(t) &= \theta_i^*(t) + \Delta_1, \    \theta_j^\text{new}(t) = \theta_j^*(t) - \Delta_1, \quad \text{for } t \in [t_1, t_1 + \delta].
\end{align*}

Let $Q^*(\cdot)$ and $Q^\text{new}(\cdot)$ denote the amount of service under the original and modified schedules, respectively. 
Consider the amount of service for job $i$ in time interval $[t_1, t_1 + \delta]$.
We have
\begin{align*}
    &Q_i^\text{new}(t_1, t_1 + \delta) - Q_i^*(t_1, t_1 + \delta)=\int_{t_1}^{t_1+\delta}\big(s(\theta_i^\text{new}(t))-s(\theta_i^*(t))\big)dt\\
    =&\int_{t_1}^{t_1+\delta}\big(s(\theta_i^*(t)+\Delta_1)-s(\theta_i^*(t))\big)dt
     =\int_{t_1}^{t_1+\delta}\big(s'(\theta_i^*(t))\Delta_1+o(\Delta_1^2)\big)dt \\
    = & s'(\theta_i^*(t))\Delta_1\delta+o(\delta^3) =s'(\theta_i^*(t_1))s'(\theta_i^*(t_2))\delta^2 + o(\delta^3)
\end{align*}
Then we have:
\begin{equation*}
    Q_i^\text{new}(t_1, t_1 + \delta) = Q_i^*(t_1, t_1 + \delta) + s'(\theta_i^*(t_1))s'(\theta_i^*(t_2))\delta^2 + o(\delta^3).
\end{equation*}
For job $j$ in time interval $[t_1, t_1 + \delta]$, similarly, we have
\begin{equation*}
     Q_j^\text{new}(t_1, t_1 + \delta) = Q_j^*(t_1, t_1 + \delta) - s'(\theta_j^*(t_1))s'(\theta_i^*(t_2))\delta^2 + o(\delta^3).
\end{equation*}

In the interval $[t_2, t_2 + \delta]$, we decrease the service rate for job $i$ by $\Delta_2 = s'(\theta_i^*(t_1))\delta$ and increase the service rate for job $j$ by the same amount:
\begin{align*}
    \theta_i^\text{new}(t) &= \theta_i^*(t) - \Delta_2, \ \theta_j^\text{new}(t) = \theta_j^*(t) + \Delta_2, \quad \text{for } t \in [t_2, t_2 + \delta].
\end{align*}
Then for job $i$ in time interval $[t_2, t_2 + \delta]$, we have
\begin{equation*}
    Q_i^\text{new}(t_2, t_2 + \delta) = Q_i^*(t_2, t_2 + \delta) - s'(\theta_i^*(t_2),)s'(\theta_i^*(t_1))\delta^2 + o(\delta^3)
\end{equation*}
For job $j$ in time interval $[t_2, t_2 + \delta]$, we have
\begin{equation*}
    Q_j^\text{new}(t_2, t_2 + \delta) = Q_j^*(t_2, t_2 + \delta) + s'(\theta_j^*(t_2))s'(\theta_i^*(t_1))\delta^2 + o(\delta^3).
\end{equation*}

Combining both intervals, we have:
\begin{align*}
    &Q_i^\text{new}(t_1, t_1 + \delta) + Q_i^\text{new}(t_2, t_2 + \delta) - Q_i^*(t_1, t_1 + \delta) - Q_i^*(t_2, t_2 + \delta) = o(\delta^3), \\
    &Q_j^\text{new}(t_1, t_1 + \delta) + Q_j^\text{new}(t_2, t_2 + \delta) - Q_j^*(t_1, t_1 + \delta) - Q_j^*(t_2, t_2 + \delta) = h_{i,j} \delta^2 + o(\delta^3).
\end{align*}

Therefore, the total service to job $i$ changes by $o(\delta^3)$, while the service to job $j$ increases by $h_{i,j} \delta^2 + o(\delta^3)$. 
Recall that $h_{i,j}>0$.
This implies that the completion time $T_i$ changes by $o(\delta^3)$, and $T_j$ decreases by a factor on the order of $\delta^2$.
That is, for a small $\delta$, the new scheduling policy $\theta_i^\text{new}(t), \theta_j^\text{new}(t)$ yields a strictly smaller objective value than the optimal $(\theta_i^*(t), \theta_j^*(t))$, contradicting to the assumption of optimality.
This completes the proof.
\end{proof}

The CDR Rule reveals a fundamental structural property of the optimal scheduling solution: the marginal benefit of allocating additional bandwidth---measured by the derivative of the speedup function---must maintain a constant ratio across any pair of active jobs over time. 
This condition ensures that no further local reallocation of resources between jobs can improve the system performance. 
The contradiction-based proof illustrates this by constructing a small bandwidth exchange between two jobs at two distinct time intervals. 
If the derivative ratio were not constant, such a reallocation would strictly reduce the objective function, violating optimality. 

Note that for the problem studied in \cite{berg2020hesrpt}, where the speedup function is $s(\theta)=\theta^p$ ($0<p<1$), the authors established a ``scale-free'' property (Theorem 3 in \cite{berg2020hesrpt}).
That is, the ratio of service rates, $\frac{\theta_i^*(t)}{\theta_j^*(t)}$ between any two active jobs $i$ and $j$ remains constant over time $t$. 
This property can be viewed as a special case of the CDR Rule: for speedup functions $s(\theta)=\theta^p$, a constant ratio of $\frac{\theta_i^*(t)}{\theta_j^*(t)}$ is equivalent to a constant ratio of $ \frac{s'(\theta_i^*(t))}{s'(\theta_j^*(t))}$.

Theorem~\ref{theorem:CDR} characterizes the case where both jobs receive positive allocations, i.e., $\theta_i^*(t) > 0$ and $\theta_j^*(t) > 0$.
We now extend the analysis to the case where one job is allocated zero resource.
Let $T_i^*$ denote the completion time of job $i$ under the optimal schedule $\theta_i^*(t)$.
For any pair of jobs $i$ and $j$, the condition $t < \min\{T_i^*, T_j^*\}$ ensures that both jobs remain incomplete under the optimal scheduler at time $t$.
The following theorem establishes the structural property for time instances where $\theta_i^*(t) > 0$ and $\theta_j^*(t) = 0$.

\begin{theorem}\label{theorem:2}
    For any pair of jobs $i,j$ and two time instances $t_1,t_2$ with $\max\{t_1,t_2\}<\min\{T_1^*,T_2^*\}$, if  $\theta_i^*(t_1) > 0$, $\theta_j^*(t_1) > 0$, $\theta_i^*(t_2) > 0$, $\theta_j^*(t_2) = 0$, then we have
    \begin{equation}
        c_{i,j}=\frac{s'(\theta_i^*(t_1))}{s'(\theta_j^*(t_1))} \leq \frac{s'(\theta_i^*(t_2))}{s'(0)},
    \end{equation}
    where $c_{i,j}$ is the constant given by Theorem~\ref{theorem:CDR}.
\end{theorem}
\begin{proof}
We prove Theorem~\ref{theorem:2} by contradiction.
Assume the theorem is false. Then we have
\begin{equation*}
    \frac{s'(\theta_i^*(t_1))}{s'(\theta_j^*(t_1))} > \frac{s'(\theta_i^*(t_2))}{s'(\theta_j^*(t_2))}.
\end{equation*}
Following the same approach in the proof of Theorem \ref{theorem:CDR}, in the interval $[t_1, t_1 + \delta]$, we increase the service rate for job $i$ by a small amount $\Delta_1 = s'(\theta_i^*(t_2))\delta$ and decrease the service rate for job $j$ by the same amount;
in the interval $[t_2, t_2 + \delta]$, we decrease the service rate for job $i$ by $\Delta_2 = s'(\theta_i^*(t_1))\delta$ and increase the service rate for job $j$ by the same amount.
We can construct a new scheduling solution $\theta_i^\text{new}(t)$ and $\theta_j^\text{new}(t)$ that decreases the objective value, which leads to a contradiction.
\end{proof}

Theorem~\ref{theorem:2} states that for any time $t$ with $t < \min\{T_i^*, T_j^*\}$, if $\theta_i^*(t) > 0$ and $\theta_j^*(t) = 0$, then  
$   \frac{s'(\theta_i^*(t))}{s'(0)} \geq c_{i,j}$.
This result can be interpreted as follows. The CDR Rule (Theorem~\ref{theorem:CDR}) requires the derivative ratio 
$   \frac{s'(\theta_i^*(t))}{s'(\theta_j^*(t))}$
to equal $c_{i,j}$ whenever both $\theta_i^*(t)$ and $\theta_j^*(t)$ are positive. Since $s'(\theta) > 0$ and $s'(\theta)$ is decreasing with $\theta$, consider a time $t$ such that $s'(\theta_i^*(t)) > c_{i,j} s'(0)$. To satisfy the CDR Rule, we would need
$s'(\theta_j^*(t)) = \frac{s'(\theta_i^*(t))}{c_{i,j}} > s'(0).$
However, because $s'(\theta)$ decreases with $\theta$, this would imply $\theta_j^*(t) < 0$, which is impossible. Thus, the non-negativity constraint enforces $\theta_j^*(t) = 0$, which leads to the condition in Theorem~\ref{theorem:2}.

Note that in \cite{berg2020hesrpt}, where $s(\theta)=\theta^p$ with $0<p<1$, we always have $\theta_i^*(t) > 0$ for any $t < T_i^*$. This occurs because $s'(0)=\infty$, implying that the benefit of allocating a small amount of resource to any unfinished job with zero scheduling rate is infinite. As a result, under the optimal scheduling, every job must receive a positive allocation until its completion.  
In contrast, in the general case considered in this paper, where $s'(0)<\infty$, it is possible (and in fact quite common) for some unfinished jobs to have zero scheduling rate. This behavior is precisely the case characterized in Theorem~\ref{theorem:2}. We will revisit this phenomenon with examples in the next section.

Theorems~\ref{theorem:CDR} and \ref{theorem:2} leads to the following corollary.

\begin{corollary2_1*}
There exists $M$ constants $c_1,c_2,\cdots,c_M$ such that
    \begin{equation}
        \frac{s'(\theta_i^*(t))}{s'(\theta_j^*(t))} = \frac{c_i}{c_j},
    \end{equation}
    for all $t$ such that $\theta_i^*(t) > 0$ and $\theta_j^*(t) > 0$.
\end{corollary2_1*}

\begin{proof}
    We prove Corollary 2.1 by induction. 
    We assume the optimal job completion time has an order $T_1^*>T_2^*>\cdots>T_N^*$\footnote{This assumption is always true for OPT (see Proposition~\ref{???}) and we use it to simplify notations. Without this assumption, the proof of Corollary 2.1 still holds after we re-order the jobs by their completion times.}.
    By definition in Theorem \ref{theorem:CDR}, it is easy to prove that for any $i$ and $j$, we have $c_{i,j}=\frac{1}{c_{j,i}}$.
    So by proving $c_{i,j}=\frac{c_i}{c_j}$ we will automatically have $c_{j,i}=\frac{c_j}{c_i}$.
    
    For jobs 1 and 2, by Theorem \ref{theorem:CDR} there exists $c_{2,1}$. We can arbitrarily set $c_1$ and set $c_2=c_{2,1}c_1$.

    For jobs $1,2,\cdots,k$ ($k\geq 2$), suppose we can find $c_1,c_2,\cdots,c_k$ such that $c_{i,j}=\frac{c_i}{c_j}$ for $i<j\leq k$.
    Then consider job $k+1$.
    If there does not exist a job $i\leq k$ and a time $t$ making $\theta_i^*(t)>0$ and $\theta_{k+1}^*(t)>0$, we can arbitrarily set $c_{k+1}$.
    If there  exists a job $i\leq k$ and a time $t$ making $\theta_i^*(t)>0$ and $\theta_{k+1}^*(t)>0$, we set $c_{k+1}=\frac{c_i}{c_{i,k+1}}$.
    All we need to prove is that for all $j\neq i$, we have $c_{k+1}=c_{k+1,j}c_j$.
    We discuss the two cases.
    At time $t$, if $\theta_j^*(t)>0$, then since $\frac{s'(\theta_i^*(t))}{s'(\theta_j^*(t))} = \frac{c_i}{c_j}$ and $\frac{s'(\theta_{k+1}^*(t))}{s'(\theta_{i}^*(t))} = \frac{c_{k+1}}{c_i}$, we have $=c_{{k+1},j}=\frac{s'(\theta_{k+1}^*(t))}{s'(\theta_{j}^*(t))} = \frac{c_{k+1}}{c_{j}}$, indicating that $c_{k+1}=c_{k+1,j}c_j$.
    If $\theta_j^*(t)=0$, then by Theorem \ref{theorem:2}, we have $\frac{c_i}{c_j}\leq \frac{s'(\theta_{i}^*(t))}{s'(0)}<1$.
    Then for any $\tau<T_{k+1}^*$ (by the order of job completion, at time $\tau$ jobs $i,j,k+1$ are not completed) such that $\theta_j^*(\tau)>0$ and $\theta_{k+1}^*(\tau)>0$, since $\frac{c_i}{c_j}<1$, we have $\theta_i^*(\tau)>0$.
    So we have $\frac{s'(\theta_i^*(\tau))}{s'(\theta_j^*(\tau))} = \frac{c_i}{c_j}$ and $\frac{s'(\theta_{k+1}^*(\tau))}{s'(\theta_{i}^*(\tau))} = \frac{c_{k+1}}{c_i}$, indicating that $c_{k+1}=c_{k+1,j}c_j$.
    This completes our proof.
\end{proof}

Compared with Theorem \ref{theorem:CDR}, Corollary 2.1 further reduces the  space of $c$'s.
We only need $M$ constants $c_i$'s instead of $\frac{M(M-1)}{2}$ constants $c_{i,j}$'s.

In summary, in this section we identify the key ``consistent derivative ratio'' property associated with the optimal solution to OPT. These results significantly reduce the search space: to find the optimal solution, it is sufficient to consider only the schedulers that satisfy these properties.  

The next question is how to construct schedulers that satisfy this property. We address this by studying the problem of determining the allocations under fixed derivative ratios in the next section. 
After that, we will present our complete solution to OPT.

\section{General Water-Filling Algorithm}\label{sec:GWF}
In this section, we study the problem of determining the resource allocation at a given time instance under fixed derivative ratios. 
The solution to this problem will serve as a fundamental component in our overall algorithm for solving OPT.  

\subsection{Constrained Allocation Problem Statement}
We refer to this problem as the \emph{Constrained Allocation Problem (CAP)}.  
Given a speedup function $s(\theta)$, a total allocation budget $0 < b \leq B$, an integer $k\geq 2$, and $k$ constants $c_1 \geq c_2 \geq \cdots \geq c_k > 0$, CAP is defined as follows:  

\begin{subequations}\label{eq:CAP}
\begin{align}
\text{CAP: }\quad  & \text{Find }  \theta_1,\theta_2,\cdots,\theta_k \tag{\ref{eq:CAP}}\\
\text{subject to} \quad & \theta_1 + \theta_2 + \cdots + \theta_k=b, \label{eq:CAP:cst1}\\
& \theta_1 \leq \theta_2 \leq \cdots \leq \theta_k, \label{eq:CAP:cst2}\\
& \frac{s'(\theta_j)}{s'(\theta_i)} = \frac{c_j}{c_i}, \quad \text{if $i<j$ and $\theta_j \geq \theta_i > 0$}, \label{eq:CAP:cst3} \\
& \frac{s'(\theta_j)}{s'(0)} \geq \frac{c_j}{c_i}, \quad \text{if $i<j$  and $\theta_j>\theta_i=0$}.\label{eq:CAP:cst4}
\end{align}
\end{subequations}

Here, constraint~\eqref{eq:CAP:cst1} enforces that the sum of allocations $\theta_1,\theta_2,\ldots,\theta_k$ equals $b$;  
constraint~\eqref{eq:CAP:cst2} requires the allocations to be ordered;  
and constraints~\eqref{eq:CAP:cst3} and \eqref{eq:CAP:cst4} restrict the ratios of the derivatives by $c_i$'s.

\subsection{Algorithm to Solve CAP}
In this subsection, we present our algorithm to solve CAP.

Let $s'^{(-1)}(y)$ denote the inverse function of $s'(\theta)$.  
The domain of $s'^{(-1)}(y)$ is $y \in [s'(b), s'(0)]$ (if $s'(0)=\infty$, then $y \in [s'(b), \infty)$).
Since $s'(\theta)$ is continuous and strictly decreasing w.r.t. $\theta$, $s'^{(-1)}(y)$ is also continuous and strictly decreasing w.r.t. $y$.
Besides, we have  $s'^{(-1)}(y)\geq 0$.

We define an \emph{auxiliary function}, denoted by $g(h)$, to help solve CAP.  
Let $[h_{\min}, h_{\max}]$ denote the domain of $g(h)$
The function $g(h)$ and $[h_{\min}, h_{\max}]$ must satisfy the following conditions:  
\begin{subequations}\label{eq:cst:gh}
\begin{align}
&\text{$g(h)$ is continuous and strictly decreasing w.r.t. $h$. } \label{eq:cst:gh_1}\\
&\text{$g(h_{\max}) \leq \frac{s'(b)}{c_1}$.} \label{eq:cst:gh_2} \\
&\text{If $s'(0)=\infty$, then $g(h_{\min})=\infty$. If $s'(0)<\infty$, then  $g(h_{\min}) \geq \frac{s'(0)}{c_k}$. \label{eq:cst:gh_3}} 
\end{align}
\end{subequations}

For each CAP problem, we need to select an auxiliary function $g(h)$ and an interval $[h_{\min}, h_{\max}]$ such that conditions~\eqref{eq:cst:gh_1}, \eqref{eq:cst:gh_2}, and \eqref{eq:cst:gh_3} are satisfied.  
Examples of auxiliary functions that work for all CAP problems include:  
$g(h)=-h$ with $h_{\min}=-\infty$, $h_{\max}=\infty$; and  
$g(h)=\tfrac{1}{h}$ with $h_{\min}=0$, $h_{\max}=\infty$.  
In fact, there exists a broad class of feasible auxiliary functions, and different choices can be made depending on the specific CAP instance.  
We will discuss the selection of $g(h)$ and $[h_{\min}, h_{\max}]$ in the following subsections.  

Given $g(h)$, for each $i=1,2,\cdots,k$, we define the function $\theta_i(h)$ as:
\begin{equation}\label{eq:theta_i_h}
    \theta_i(h)=\begin{cases}
    0 & \text{if } c_ig(h)\geq s'(0), \\
    s'^{(-1)}\big(c_ig(h)\big) & \text{if } s'(b)<c_ig(h)<s'(0), \\
   b & \text{if } c_ig(h) \leq s'(b).
\end{cases}
\end{equation}
It can be shown that for each $i=1,2,\ldots,k$, $\theta_i(h)$ is continuous and increasing in $h \in [h_{\min}, h_{\max}]$.  
We refer to $\theta_i(h)$ as the \emph{water-filling function} for job $i$, since it can be interpreted as the ``water volume in bottle $i$ when the water level is $h$.''  
This analogy will be further discussed in Section~\ref{subsec:water_filling_illustration}.

Let $\beta(h)$ denote the sum of $\theta_i(h)$'s across all jobs: 
\begin{equation}\label{eq:beta_h}
    \beta(h) = \theta_1(h) + \theta_2(h) + \cdots + \theta_k(h).
\end{equation}
In the water-filling analogy of Section~\ref{subsec:water_filling_illustration}, $\beta(h)$ represents the ``total water volume across all bottles when the water level is $h$.''  

We consider the following problem, which is referred as \emph{Water-Filling Problem (WFP)}.
\begin{subequations}\label{eq:WFP}
\begin{align}
\text{WFP: }\quad  & \text{Find }  h \tag{\ref{eq:WFP}}\\
\text{subject to} \quad & \beta(h)=b. \label{eq:WFP:cst1}
\end{align}
\end{subequations}

The following proposition states the existence of a solution to WFP.
\begin{proposition}\label{prop:solution:WFP}
    WFP has a unique solution.
\end{proposition}

\begin{proof}
    It can be shown that $\beta(h)$ is continuous and increasing.
    Since $\beta(h_{\min})=0<b$ and $\beta(h_{\max})=kb> b$, there exists a unique $h$ making $\beta(h)=b$. 
    This completes our proof.
\end{proof}

The following proposition states the a solution of WFP is also a solution to CAP.
\begin{proposition}\label{prop:solution:WFP_CAP}
    If $h$ is a solution to WFP, then $\theta_i=\theta_i(h)$ for $i=1,2,\cdots,k$ is a solution to CAP.
\end{proposition}

\begin{proof}
    We need to verify this solution $\theta_i=\theta_i(h)$ satisfies all the constraints in CAP.

    By \eqref{eq:WFP:cst1}, we have $\theta_1 + \theta_2 + \cdots + \theta_k=b$. So the solution satisfies \eqref{eq:CAP:cst1}.

   It can be shown that $g(h)>0$ (otherwise $\theta_i(h)=b$ for all $i$'s and \eqref{eq:WFP:cst1} cannot be satisfied).
   
   For any $i<j$, we have $c_ig(h)\geq c_jg(h)$.
   By definition in \eqref{eq:theta_i_h}, we have $\theta_i(h)\leq \theta_j(h)$.
   So the solution satisfies \eqref{eq:CAP:cst2}.

   Consider any $i<j$ and $\theta_i(h)>0$.
   We have $0<\theta_i(h)<\theta_j(h)<b$.
   And we have $\theta_i(h)= s'^{(-1)}\big(c_ig(h)\big)$ and $\theta_j(h)= s'^{(-1)}\big(c_jg(h)\big)$.
   And we have $ \frac{s'(\theta_j)}{s'(\theta_i)} = \frac{c_jg(h)}{c_ig(h)}=\frac{c_j}{c_i}$.
   So the solution satisfies \eqref{eq:CAP:cst3}.

   Consider any $i<j$, $\theta_i(h)=0$, and $\theta_j(h)>0$.
   We have $c_ig(h)\geq s'(0)$ and $\theta_j(h)= s'^{(-1)}\big(c_jg(h)\big)$.
   And we have $ \frac{s'(\theta_j)}{s'(\theta_i)} = \frac{c_jg(h)}{s'(0)} \geq \frac{c_jg(h)}{c_ig(h)}=\frac{c_j}{c_i}$.
   So the solution satisfies \eqref{eq:CAP:cst4}.
   This completes our proof.
\end{proof}

The following proposition states the number of solutions to CAP.

\begin{proposition}\label{proposition:solution CAP}
    CAP has at most one solution.  
\end{proposition}

\begin{proof}
    We prove Proposition~\ref{proposition:solution CAP} by contradiction.
    Suppose CAP has two different solutions $\theta_1,\theta_2,\cdots,\theta_k$ and 
    $\hat{\theta}_1,\hat{\theta}_2,\cdots,\hat{\theta}_k$.
    Since $\theta_1 + \theta_2 + \cdots + \theta_k =\hat{\theta}_1 + \hat{\theta}_2 + \cdots + \hat{\theta}_k=b$, without loss of generality, we can find $i<j$ making $\hat{\theta}_i<\theta_i \leq \theta_j<\hat{\theta}_j$.
    Based on constraints~\eqref{eq:CAP:cst3} and \eqref{eq:CAP:cst4}, we have
    \begin{equation}\label{eq:proof:prop1}
      \frac{s'(\hat{\theta}_j)}{s'(\hat{\theta}_i)}\geq  \frac{c_j}{c_i}= \frac{s'(\theta_j)}{s'(\theta_i)}. 
    \end{equation}
    On the other hand, since $s'(\theta)$ is strictly decreasing and positive, we have $0<s'(\hat{\theta}_j)<s'(\theta_j)$ and $0<s'(\theta_i)<s'(\hat{\theta}_i)$.
    So we have
    \begin{equation}\label{eq:proof:prop2}
      \frac{s'(\hat{\theta}_j)}{s'(\hat{\theta}_i)}< \frac{s'(\theta_j)}{s'(\theta_i)}. 
    \end{equation}
    We can see \eqref{eq:proof:prop1} and \eqref{eq:proof:prop2} leads to a contradiction.
    This completes our proof.
\end{proof}

Combining Propositions \ref{prop:solution:WFP}, \ref{prop:solution:WFP_CAP}, and \ref{proposition:solution CAP}, we have the following theorem.
\begin{theorem}\label{theoren:CAP_WFP}
    CAP has a unique solution, which is the solution of WFP.
\end{theorem}

\begin{algorithm}
\caption{General Water-Filling (GWF) Algorithm to Solve CAP}\label{alg:CAP}
\begin{algorithmic}[1] 
\REQUIRE   $s(\theta)$, $b$, $k$, and $c_1, c_2 , \cdots, c_k $.
\ENSURE  The solution to CAP: $ \theta_1^* , \theta_2^* , \cdots , \theta_k^*$.
\STATE Select an auxiliary functions $g(h)$ that satisfies \eqref{eq:cst:gh_1}, \eqref{eq:cst:gh_2}, and \eqref{eq:cst:gh_3}.
\STATE Solve WFP and get its solution $h^*$.
\STATE Compute $\theta_i^*=\theta_i(h^*)$ based on \eqref{eq:theta_i_h} for $i=1,2,\cdots,k$.
\end{algorithmic}
\end{algorithm}

Based on Theorem \ref{theoren:CAP_WFP}, we propose the \emph{General Water Filling (GWF)} algorithm to solve CAP, as shown in Algorithm~\ref{alg:CAP}.

Note that in Step 1 of Algorithm~\ref{alg:CAP}, we can select $g(h)$ from a large class of feasible auxiliary functions.
However, a good selection of $g(h)$ can make WFP much easier to solve, i.e., make Step 2  of Algorithm~\ref{alg:CAP} easier.
We will use an example to show this point in the next subsection.

\subsection{An Example for Selecting $g(h)$}
Consider a speedup function: 
$s(\theta)=a(\theta+z)^p-az^p$, where $a>0$, $z\geq 0$, and $0<p<1$ are constants.
Note that when $z=0$, the speedup function $s(\theta)=a\theta^p$ is the one considered in \cite{berg2020hesrpt}.
We will use this speedup function as an example to show how to select $g(h)$ to make WFP easier to solve.

For this speedup function, we have $s'(\theta)=ap(\theta+z)^{p-1}$ and $s'^{(-1)}(y)=(\frac{y}{ap})^{\frac{1}{p-1}}-z$.
Consider \eqref{eq:theta_i_h}.
When $s'(b)<c_ig(h)<s'(0)$, for each $i=1,2,\cdots, k$ we have 
\begin{equation}\label{eq:example1}
    \theta_i(h)=s'^{(-1)}\big(c_ig(h)\big)=(\frac{c_ig(h)}{ap})^{\frac{1}{p-1}}-z.
\end{equation}

If we select $g(h)=aph^{p-1}$ and $h_{\min}=0$, $h_{\max}=\infty$, then \eqref{eq:example1} becomes:
\begin{equation*}
    \theta_i(h)=c_i^{\frac{1}{p-1}}h-z=c_i^{\frac{1}{p-1}}(h-zc_i^{\frac{1}{1-p}}),
\end{equation*}
and for each $i=1,2,\cdots, k$ we have
\begin{equation}\label{eq:example:WF1}
    \theta_i(h)=\begin{cases}
    0 & \text{if } h \leq zc_i^{\frac{1}{1-p}}, \\
    c_i^{\frac{1}{p-1}}(h-zc_i^{\frac{1}{1-p}}) & \text{if } zc_i^{\frac{1}{1-p}}<h<(z+b)c_i^{\frac{1}{1-p}}, \\
   b & \text{if }h \geq (z+b)c_i^{\frac{1}{1-p}}.
\end{cases}
\end{equation}
We can see the selection of $g(h)=aph^{p-1}$ will give us a piece-wise linear $\theta_i(h)$ for all $i$'s, making WFP  easy to solve.
This means  $g(h)=aph^{p-1}$ is a good selection.

On the other hand, if we select $g(h)=ap\log(1+\frac{1}{h})$ and $h_{\min}=0$, $h_{\max}=\infty$, then \eqref{eq:example1} becomes:
and for each $i=1,2,\cdots, k$ we have
\begin{equation}\label{eq:example:WF2}
    \theta_i(h)=\begin{cases}
    0 & \text{if } c_ig(h)\geq s'(0), \\
   \big(c_i\log(1+\frac{1}{h})\big)^{\frac{1}{p-1}}-z & \text{if } s'(b)<c_ig(h)<s'(0), \\
   b & \text{if } c_ig(h) \leq s'(b).
\end{cases}
\end{equation}
This is non-linear and make WFP not easy to solve. This means $g(h)=ap\log(1+\frac{1}{h})$ is not a good selection.

\subsection{Water-Filling Illustration} \label{subsec:water_filling_illustration}
Water-filling algorithm is a classic algorithm for solving the power allocation problem in multi-channel communication systems.  
As illustrated in Fig.~5.11 in \cite{tse2005fundamentals}, the algorithm can be visualized as pouring ``water'' (i.e., power) into parallel subcarriers such that the water level across all subcarriers containing water is equalized. Each subcarrier has the same ``width'' but a different ``bottom height,'' determined by its channel quality.  
The water-filling algorithm is proven to be optimal for power allocation in multi-channel communication. Further details can be found in Chapter~5.3.3 of \cite{tse2005fundamentals}.

\begin{figure}
  \centering
  \includegraphics[width=0.8\textwidth]{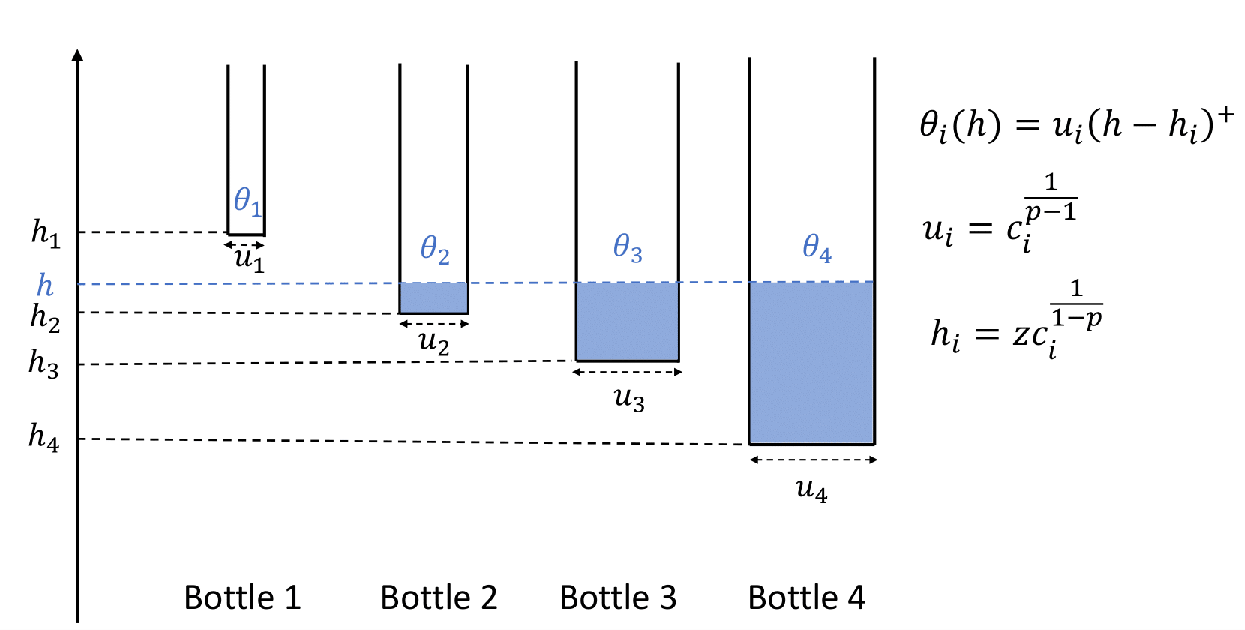}
  \caption{An illustration for GWF Algorithm when $s(\theta)=a(\theta+z)^p-az^p$ and $k=4$.}
  \label{fig:WF1} 
\end{figure}

We observe that in our GWF algorithm for solving SAP, the allocation result can be naturally interpreted through the water-filling analogy.  
Fig.~\ref{fig:WF1} illustrates the GWF algorithm when $s(\theta)=a(\theta+z)^p-az^p$ and $k=4$.  
In this case, we construct $k=4$  ``rectangle'' bottles, each with its own ``width'' and ``bottom height.''  

Let $u_i$ and $h_i$ denote the width and bottom height of bottle $i$, respectively, and let $\theta_i(h)$ denote the volume of water in bottle $i$ when the water level is $h$.  
We set $u_i = c_i^{\tfrac{1}{p-1}}$ and $h_i = zc_i^{\tfrac{1}{1-p}}$.  

Define $(x)^+ = \max(0,x)$.  
From \eqref{eq:example:WF1}, we have $\theta_i(h) = u_i(h-h_i)^+$ whenever $\theta_i(h) < b$.  
Thus, solving WFP reduces to finding the water level $h$ in Fig.~\ref{fig:WF1} such that the total water volume across all bottles equals $b$.  

This is directly analogous to the classical water-filling algorithm in multi-channel communication  with one important generalization:  
in GWF (Fig.~\ref{fig:WF1}), the bottle widths differ across jobs, whereas in the classical water-filling algorithm, all channel widths are identical.  
Note that when $z=0$, all bottles have the same bottom heights, which indicates for any $b>0$, $\theta_i^*>0$ for all $i$'s.

\begin{figure}
  \centering
  \includegraphics[width=0.8\textwidth]{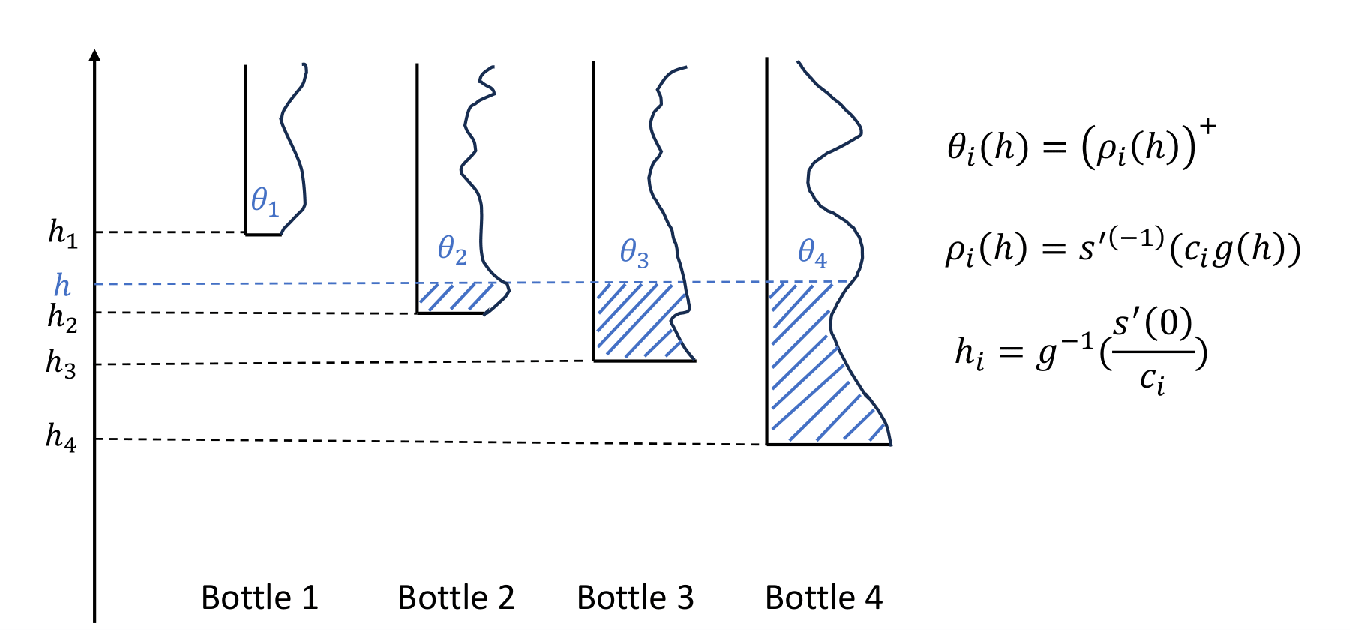}
  \caption{An illustration for GWF Algorithm for non-rectangle bottles.}
  \label{fig:WF2} 
\end{figure}

A more general illustration for GWF is shown in Fig.~\ref{fig:WF2}, where the bottles are not rectangle-shape.
Here we define $\theta_i(h)=(\rho_i(h))^+$, where $\rho_i(h)=s'^{(-1)}(c_ig(h))$.
Let $g^{-1}(\cdot)$ denote the inverse function of $g(h)$.
And we define $h_i=g^{-1}(\frac{s'(0)}{c_i})$.
We can see $\theta_i(h)$ in \eqref{eq:theta_i_h} can be illustrated as filling water in bottles, as shown in Fig.~\ref{fig:WF2}.
Note that when $s'(0)=\infty$, we have $h_1=h_2=\cdots=h_k$, indicating that all bottles have the same bottom height.

\subsection{Solving WFP}
In this subsection, we discuss how to solve WFP, i.e., how to run Step 2 in Algorithm \ref{alg:CAP}.

\subsubsection{Closed-Form Solution}
In the ideal case, we can get a closed-form solution to WFP.
In this paper we identify a class of speedup functions $s(\theta)$, which we call \emph{regular functions}.
Under regular $s(\theta)$, we can get a closed-form solution to WFP, and thus get a  closed-form solution to CAP by Algorithm \ref{alg:CAP}.

\begin{definition}\label{def:regular}
    A speedup function $s(\theta)$ is regular if there exists $\alpha\neq 0,\gamma\neq 0$ and $z\in \mathcal{R}$ such that $s'(\theta)=\alpha(\theta+z)^\gamma$ for all $\theta\in[0,B]$.
\end{definition}

\begin{table}
\small
\caption{Some examples of regular speedup functions}\label{tab:reg1}
\centering
\begin{tabular}{|c|c|}
\hline Regular Speedup Function & Example \\
\hline
$s(\theta)=a(\theta+z)^p-az^p, \theta\in[0,B]$, where $a>0$, $z\geq 0$, and $0<p<1$ & $s(\theta)=(\theta+1)^{0.5}$ \\
\hline
$s(\theta)=a\ln (p\theta+1), \theta\in[0,B]$, where $a>0$ and $p>0$ & $s(\theta)=\ln (\theta+1)$ \\
\hline
$s(\theta)=az^p-a(\theta+z)^p, \theta\in[0,B]$, where $a>0$, $z\geq 0$, and $p<0$ & $s(\theta)=\frac{\theta}{\theta+1}$ \\
\hline
$s(\theta)=az^p-a(z-\theta)^p, \theta\in[0,B]$, where $a>0$, $ z\geq B$, and $p>1$ & $s(\theta)=2\theta-\theta^2$ ($B\leq 1$) \\
\hline
\end{tabular}
\end{table}

Some examples of regular speedup functions are given in Table \ref{tab:reg1}.
Readers can verify each row in Table \ref{tab:reg1} satisfies Definition \ref{def:regular}.
We can see the speedup function $s(\theta)=a\theta^p$ ($0<p<1)$ considered in \cite{berg2020hesrpt} is regular.

For regular $s(\theta)$ such that $s'(\theta)=\alpha(\theta+z)^\gamma$, we have a closed-form expression for $s'^{(-1)}(y)$:
\begin{equation*}
    s'^{(-1)}(y)=(\frac{y}{\alpha})^{\frac{1}{\gamma}}-z.
\end{equation*}
We can select $g(h)=(\alpha h)^\gamma$.
Based on \ref{eq:theta_i_h}, we have
$\theta_i(h) = u_i(h-h_i)^+$ whenever $\theta_i(h) < b$, where $u_i = c_i^{\tfrac{1}{\gamma}}$ and $h_i = z c_i^{-\tfrac{1}{\gamma}}$.  
In the water-filling illustration, the bottles are rectangles.
And the problem WFP is:
\begin{subequations}
\begin{align*}
\text{WFP under regular $s(\theta)$:}\quad  & \text{Find }  h \\
& \text{subject to} \quad  \beta(h)=\sum_{i=1}^ku_i(h-h_i)^+=b. 
\end{align*}
\end{subequations}
Notice that $\beta(h)$ is a piecewise linear function.
WFP under regular $s(\theta)$ can be solved easily. 
We omit the details here.

For general non-regular $s(\theta)$, whether we can select a good $g(h)$ to get an closed-form expression of the solution to WFP is unclear to us.
We believe for functions like $s(\theta)=x\theta^\frac{1}{2}+\ln(1+\theta)$, it is hard to write down a closed-form expression for $s'^{(-1)}(y)$, and thus it is very difficult to get an closed-form solution to WFP.

\subsubsection{Numerical Methods}
For non-regular speedup functions that we don't have closed-form solutions, we can always use numerical methods to get a numerical solution.
For example, we can arbitrarily select a feasible $g(h)$ (e.g., select $g(h)=-h$ with $h_{\min}=-\infty$, $h_{\max}=\infty$), and numerically compute $\beta(h)$ by \eqref{eq:beta_h}.
When $\beta(h)<b$, we increase $h$;
when $\beta(h)>b$, we decrease $h$, until the gap between $\beta(h)$ and $b$ is small enough.
This can be done by methods like bisection \cite{sikorski1982bisection} to solve it.
An alternative efficient way to solve WFP numerically will be introduced in Section \ref{???}.

\section{SmartFill Algorithm}\label{sec:SmartFill}
Based on the CDR Rule and the GWF algorithm, in this section we present the \emph{SmartFill} Algorithm, which is shown to be able to solve OPT optimally.
When $s(\theta)$ is regular, SmartFill provides a closed-form solution.
When $s(\theta)$ is non-regular, SmartFill provides a numerical method.

\subsection{Preliminary Results}
Before presenting SmartFill, we introduce some preliminary results to OPT. 
Many of them directly come from \cite{berg2020hesrpt}.

Let $m^*(t)$ denote the number of jobs in the system at time $t$ under the optimal scheduling policy. We begin with the following proposition:
\begin{proposition}\label{prop:7}
    For any two time instances $t_1$ and $t_2$, if $m^*(t_1) = m^*(t_2)$, then $\theta_i^*(t_1) = \theta_i^*(t_2)$ for all $i = 1, 2, \cdots, M$.
\end{proposition}
Proposition~\ref{prop:7} is adapted from Theorem 1 in \cite{berg2020hesrpt}. 
The proof of Proposition~\ref{prop:7} is identical to the proof of Theorem 1 in \cite{berg2020hesrpt}.

Regarding the job completion order, we have the following proposition:
\begin{proposition}\label{prop:8}
    The optimal policy completes jobs in Shortest-Job-First (SJF) order: $M, M-1, \cdots, 1$.
\end{proposition}
The proof of Proposition~\ref{prop:8} is identical to the proof of Theorem 2 in \cite{berg2020hesrpt}.
Let $T_i^*$ denote the completion time for job $i$ ($i=1,2,\cdots,M)$ under the optimal policy.
By Proposition~\ref{prop:8}, we have
\begin{equation}
    T_1^*>T_2^*>\cdots>T_M^*.
\end{equation}

Based on Propositions \ref{prop:7} and \ref{prop:8}, let $\theta_i^j$ denote the constant scheduling rate allocated to job $i$ during the time interval $[T_{j+1}^*, T_j^*)$, where we define $T_{M+1}^* = 0$ for convenience.
We then define the scheduling matrix $\Theta=[\theta_i^j]$ as
\begin{equation}
    \Theta = 
    \begin{bmatrix}
        \theta_1^1 & \theta_1^2 & \cdots & \theta_1^M \\
        \theta_2^1 & \theta_2^2 & \cdots & \theta_2^M \\
        \vdots     & \vdots     & \ddots & \vdots     \\
        \theta_M^1 & \theta_M^2 & \cdots & \theta_M^M \\
    \end{bmatrix}.
\end{equation}
Clearly, finding the optimal scheduling policy is equivalent to determining the optimal matrix $\Theta$.

The resource constraint in \eqref{eq:cst1} translates into the following condition on the scheduling matrix:
\begin{equation}\label{eq:cst4}
    \sum_{i=1}^M \theta_i^j = B, \quad \forall j = 1, 2, \cdots, M.
\end{equation}
In \eqref{eq:cst4}, we use equality rather than inequality because, under the optimal policy, the total available bandwidth should be fully utilized at all times.

Since the optimal scheduler follows the SJF order, we have $\theta_i^j = 0$ for all $i>j$. In other words, the scheduling matrix $\Theta$ is upper triangular:
\begin{equation}
    \Theta = 
    \begin{bmatrix}
        \theta_1^1 & \theta_1^2 & \cdots & \theta_1^M \\
        0          & \theta_2^2 & \cdots & \theta_2^M \\
        \vdots     & \vdots     & \ddots & \vdots     \\
        0          & 0          & \cdots & \theta_M^M \\
    \end{bmatrix}.
\end{equation}

\subsection{Main Algorithm}
The idea of SmartFill is to iteratively compute the optimal scheduling results from the last completed job (job 1) to the first completed job (job $M$).
After each iteration $k$, we record the $c_i$'s  $(1\leq i\leq k)$ for the CDR rule, and use them to compute the optimal scheduling results for iteration $k+1$.

For each $b\leq B$ and $1\leq i\leq k$, we define the CAP function
\begin{equation}
    \text{CAP}_i(b,c_1,c_2,\cdots,c_k)
\end{equation}
as the optimal solution $\theta_i $ for the CAP problem.
The value of CAP function can be computed by the GWF algorithm (Algorithm \ref{alg:CAP}).

The following propositions states that the optimal objective is a linear combination of $x_i$'s.

\begin{proposition}\label{prop:9}
    Under fixed $w_i$'s, $s(\theta)$, and $B$, for different $x_i$'s such that $x_1 \geq x_2 \geq \cdots \geq x_M$, the optimal objective value can be expressed as
    \begin{equation}
        J^* = \sum_{i=1}^M a_i x_i,
    \end{equation}
    where the coefficients satisfy $a_1 < a_2 < \cdots < a_M$, and $a_i$'s are independent of the specific values of $x_i$.
\end{proposition}
\begin{proofsketch}
    We prove Proposition \ref{prop:9} by induction.
    On one hand, it is apparent that Proposition \ref{prop:9} holds for $M=1$.

    On the other hand, assume Proposition \ref{prop:9} holds for $M=k$ ($k\geq 1$).
    Then consider the case $M=k+1$.
    Consider the time instance $T_M$ when job $M$ completes.
    After $T_M$, the optimal weighted sum of completion times can be wrote as $J'=\sum_{i=1}^k\alpha_ix_i'$, where $x_i'$ is the remaining size for job $i$ at $t=T_M$.
    And the parameters $c_1,c_2,\cdots,c_k$ can be determined.
    Then for time interval $[0, T_M]$, when the service rate for job $M=k+1$ is $\theta_M^M=\mu$, the service rate of  for job $i$ ($i\leq k$) is $\theta_i^M=\text{CAP}_i(B-\mu,c_1,c_2,\cdots,c_k)$.
    The optimal solution of $\mu$ is given by the following problem:
\begin{equation*}
    \arg\max_\mu \ \  (\sum_{i=1}^{k+1}w_i)\cdot\frac{x_{k+1}}{s(\mu)}-  \big(\sum_{i=1}^{k}a_is(\text{CAP}_i(B-\mu,c_1,c_2,\cdots,c_k))\big)\cdot\frac{x_{k+1}}{s(\mu)}.
\end{equation*}
It is easy to see that the value of $\mu$ is independent from the value of $x_{k+1}$. 
So $J^*$ is also linear with $x_{k+1}$.
This completes the proof.
\qed    
\end{proofsketch}
    
Now we present the SmartFill Algorithm. 
For iteration $k=1$, we set $\theta_1^1=B$, $c_1=1$, and $a_1=\frac{w_1}{s(B)}$.
After iteration $k$, we will have $c_1,c_2,\cdots,c_k$ and $a_1,a_2,\cdots,a_k$.
For iteration $k+1$, we compute $\theta_{k+1}^{k+1}$ by 
\begin{equation}\label{eq:SF_1}
    \theta_{k+1}^{k+1}= \arg\max_\mu \ \  \frac{\sum_{i=1}^{k+1}w_i- \sum_{i=1}^{k}a_is(\text{CAP}_i(B-\mu,c_1,c_2,\cdots,c_k))}{s(\mu)}.
\end{equation}
Then we compute $\theta_{i}^{k+1}$ for $i=1,2,\cdots,k$ by
\begin{equation}\label{eq:SF_2}
    \theta_{k+1}^{i}= \text{CAP}_i(B-\theta_{k+1}^{k+1},c_1,c_2,\cdots,c_k).
\end{equation}
Then we compute $c_{k+1}$ and $a_{k+1}$ by:
\begin{equation}\label{eq:SF_3}
    c_{k+1}= \frac{s'(\theta_{k+1}^{k+1})}{s'(\theta_{k}^{k+1})} c_k,
\end{equation}
and 
\begin{equation}\label{eq:SF_4}
    a_{k+1}=   \frac{\sum_{i=1}^{k+1}w_i- \sum_{i=1}^{k}a_is( \theta_{i}^{k+1})}{s(\theta_{k+1}^{k+1})}.
\end{equation}

\begin{algorithm}
\caption{SmartFill Algorithm to Solve OPT}\label{alg:SF}
\begin{algorithmic}[1] 
\REQUIRE   $s(\theta)$, $B$, $M$, $x_i$'s, and $w_i$'s.
\ENSURE  The optimal solution $\Theta=[\theta_i^j]$ to OPT.
\STATE Set $\theta_i^j = 0$ for all $i>j$.
\STATE  Set $\theta_1^1=B$, $c_1=1$, and $a_1=\frac{w_1}{s(B)}$.
\FOR{$k=1,2,\cdots,M-1$}
\STATE Compute $\theta_{k+1}^{k+1}$ by \eqref{eq:SF_1}.
\STATE Compute $\theta_{i}^{k+1}$ for $i=1,2,\cdots,k$ by \eqref{eq:SF_2}.
\STATE Compute $c_{k+1}$ and $a_{k+1}$ by \eqref{eq:SF_3} and \eqref{eq:SF_4}.
\ENDFOR
\end{algorithmic}
\end{algorithm}

The complete procedure of SmartFill is given in Algorithm \ref{alg:SF}.
Note that in Step 4, we need to solve a maximization problem defined in \eqref{eq:SF_1}.
When $s(x)$ is regular, by GEF, $\text{CAP}_i(B-\mu,c_1,c_2,\cdots,c_k)$ is linear with $\mu$ for each $i$, and we can get a closed-form expression of the solution.
When $s(x)$ is non-regular, we can do this step numerically based on GWF.

\begin{figure}
\centering
\begin{minipage}{.5\textwidth}
  \centering
  \includegraphics[width=.9\linewidth]{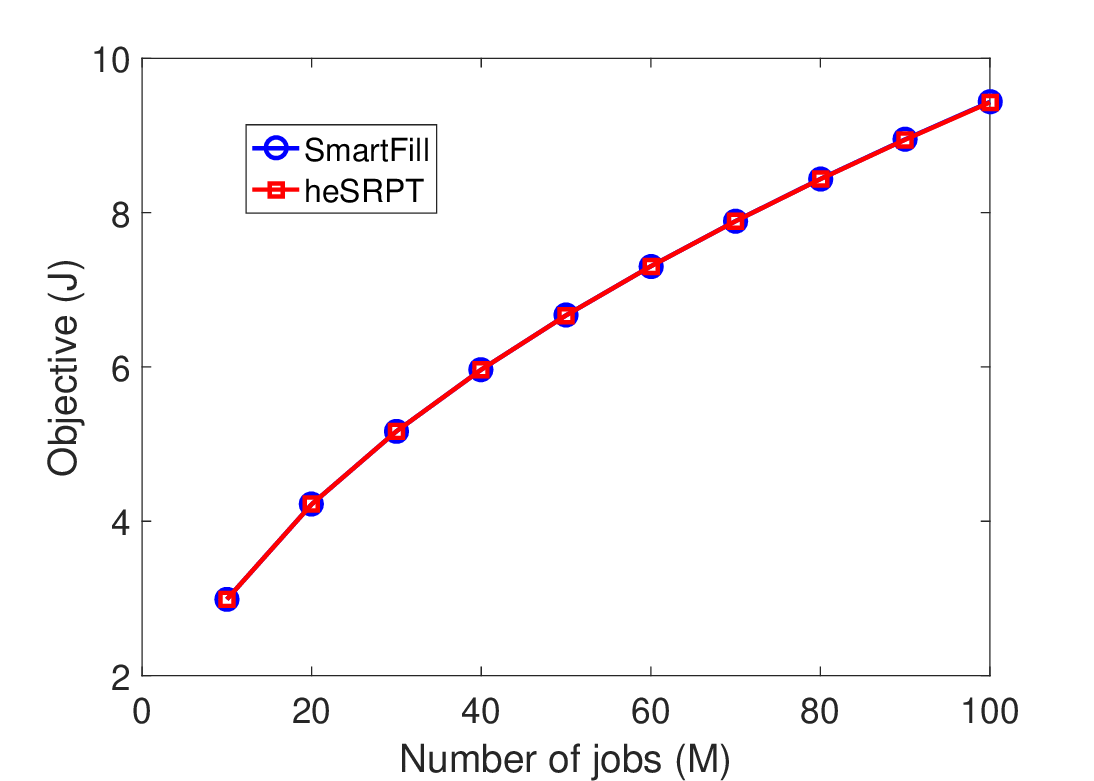}
  \captionof{figure}{Mean slowdown when $s(\theta)=\theta^{0.5}$.}
  \label{fig:3_1}
\end{minipage}%
\begin{minipage}{.5\textwidth}
  \centering
  \includegraphics[width=.9\linewidth]{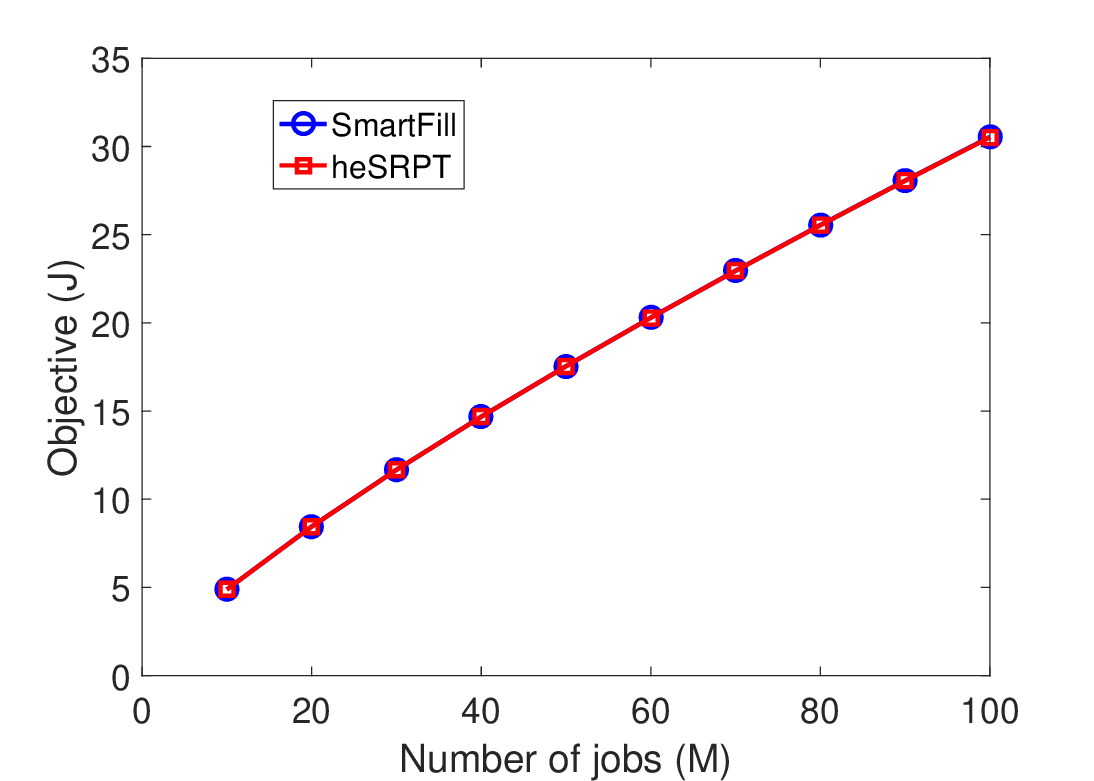}
  \captionof{figure}{Mean slowdown when $s(\theta)=10\theta^{0.8}$.}
  \label{fig:3_2}
\end{minipage}
\end{figure}

\section{Numerical Results}\label{sec:simulation}
In this section we provide some numerical results for SmartFill.
We use heSRPT \cite{berg2020hesrpt} as the performance benchmark.

\subsection{Special Case: $s(\theta)=a\theta^p$}
We first consider the special case $s(\theta)=a\theta^p$, where  heSRPT is optimal \cite{berg2020hesrpt}. 
We consider $M=10,20,\cdots,100$ jobs.
We consider $B=10$.
For each $M$, we assume $x_1=M,x_2=M-1,\cdots,x_M=1$.
We consider the mean slowdown at the objective $J$, i.e.,  $w_i=\frac{1}{x_i}$.
Figs.~\ref{fig:3_1} and \ref{fig:3_2} shows the mean slowdown under different $M$'s when $s(\theta)=\theta^{0.5}$ and $s(\theta)=10\theta^{0.8}$, respectively.
In all cases, SmartFill achieves the same mean slowdown as heSRPT, thereby validating that SmartFill is optimal in this setting.

\subsection{General Concave $s(\theta)$}
We next evaluate the performance of SmartFill under general concave speedup functions.  
In \cite{berg2020hesrpt}, the authors suggested that when $s(\theta)\neq a\theta^p$, the heSRPT policy can still be applied by approximating $s(\theta)$ with a function of the form $a\theta^p$.  
We adopt this approximation-based heSRPT as our performance benchmark.  

As before, we consider systems with $M=10,20,\ldots,100$ jobs and a total resource budget $B=10$.  
For each $M$, the job sizes are set as $x_1=M, x_2=M-1, \ldots, x_M=1$.  
The performance metric is again the mean slowdown, defined by the objective $J$ with $w_i=\tfrac{1}{x_i}$.  

Fig.~\ref{fig:3_3_2} shows the mean slowdowns of SmartFill and heSRPT when $s(\theta)=\log(1+\theta)$.  
Here, heSRPT uses the approximation $s(\theta)=0.79\theta^{0.48}$, shown in Fig.~\ref{fig:3_3_1}.  
SmartFill consistently outperforms heSRPT, achieving a 13.6\% lower slowdown at $M=100$.  

Fig.~\ref{fig:3_4_2} shows results for $s(\theta)=\sqrt{4+\theta}-2$, where heSRPT employs the approximation $s(\theta)=0.26\theta^{0.82}$, as shown in Fig.~\ref{fig:3_4_1}.  
SmartFill again consistently outperforms heSRPT.  
Because this approximation is tighter than in the previous case, the performance gap is smaller: at $M=100$, SmartFill’s slowdown is 6.3\% lower than that of heSRPT.

\begin{figure}
\centering
\begin{minipage}{.5\textwidth}
  \centering
  \includegraphics[width=.9\linewidth]{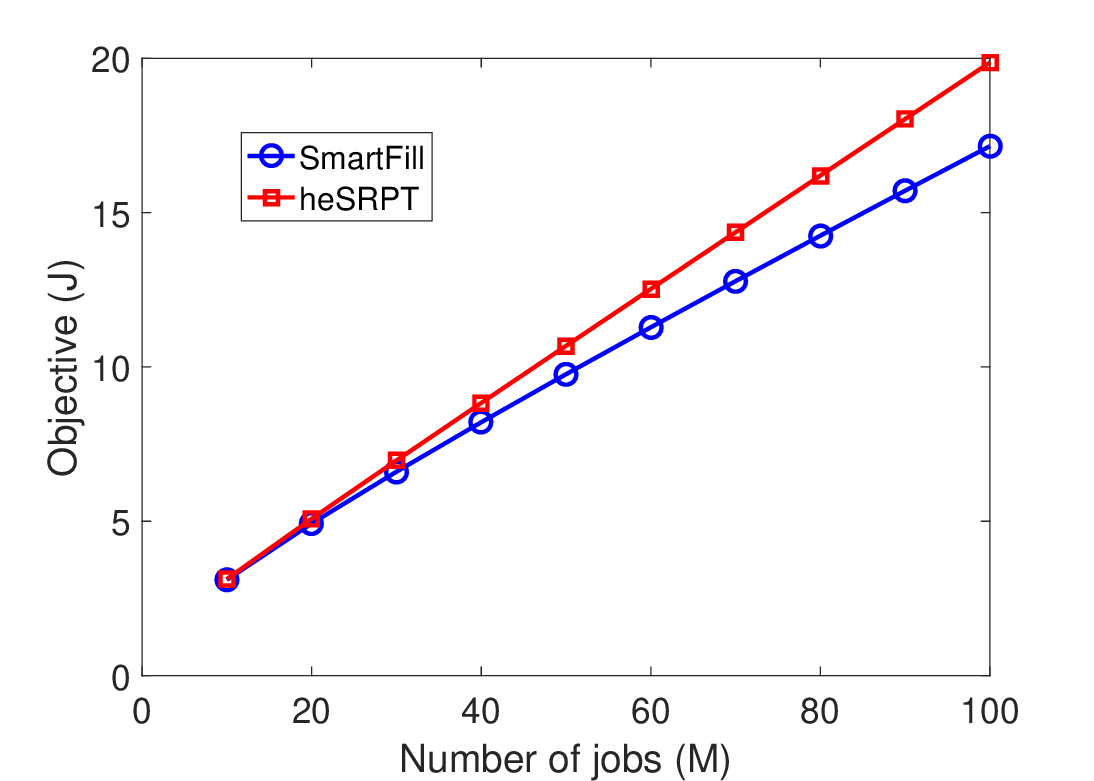}
  \captionof{figure}{Mean slowdown when $s(\theta)=\log(1+\theta)$.}
  \label{fig:3_3_2}
\end{minipage}%
\begin{minipage}{.5\textwidth}
  \centering
  \includegraphics[width=.9\linewidth]{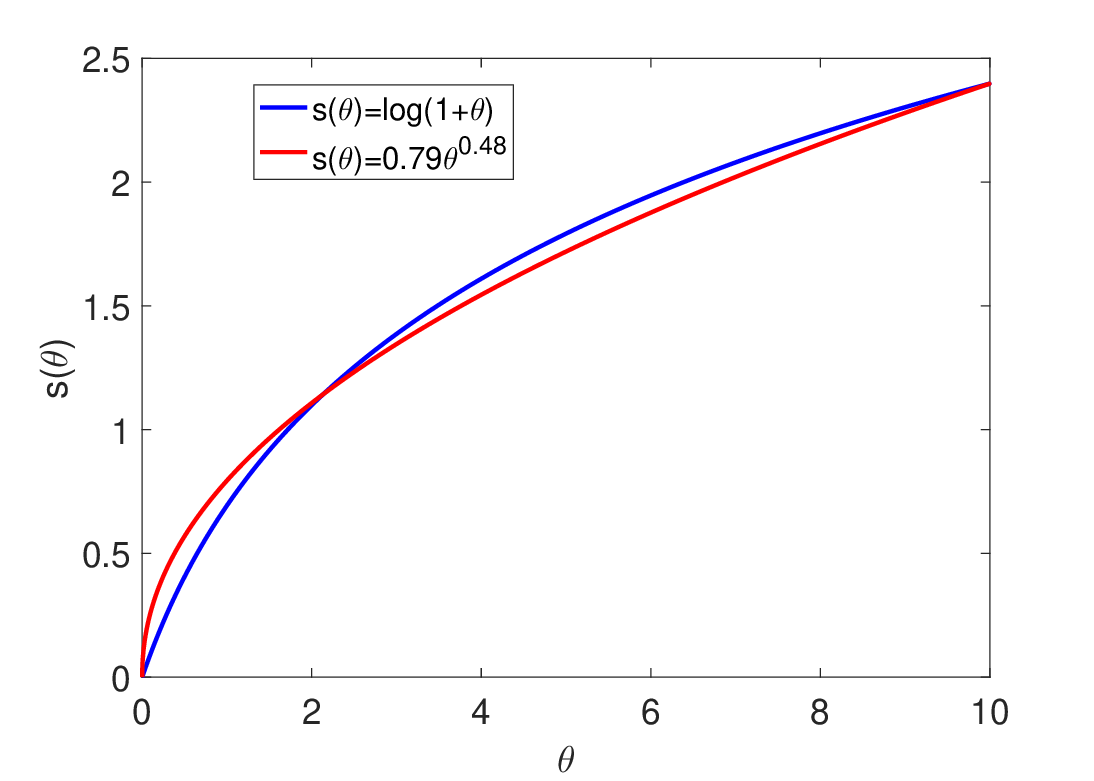}
  \captionof{figure}{Approximation for $s(\theta)=\log(1+\theta)$ used by heSRPT.}
  \label{fig:3_3_1}
\end{minipage}
\end{figure}

\begin{figure}
\centering
\begin{minipage}{.5\textwidth}
  \centering
  \includegraphics[width=.9\linewidth]{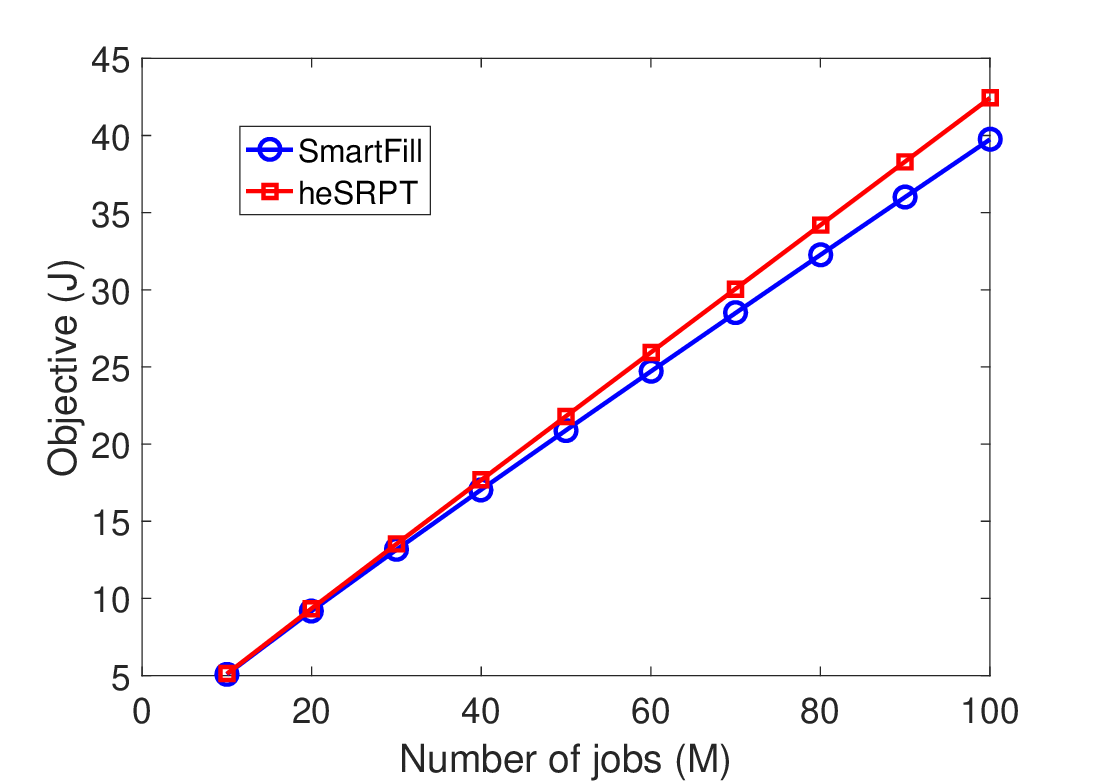}
  \captionof{figure}{Mean slowdown when $s(\theta)=\sqrt{4+\theta}-2$.}
  \label{fig:3_4_2}
\end{minipage}%
\begin{minipage}{.5\textwidth}
  \centering
  \includegraphics[width=.9\linewidth]{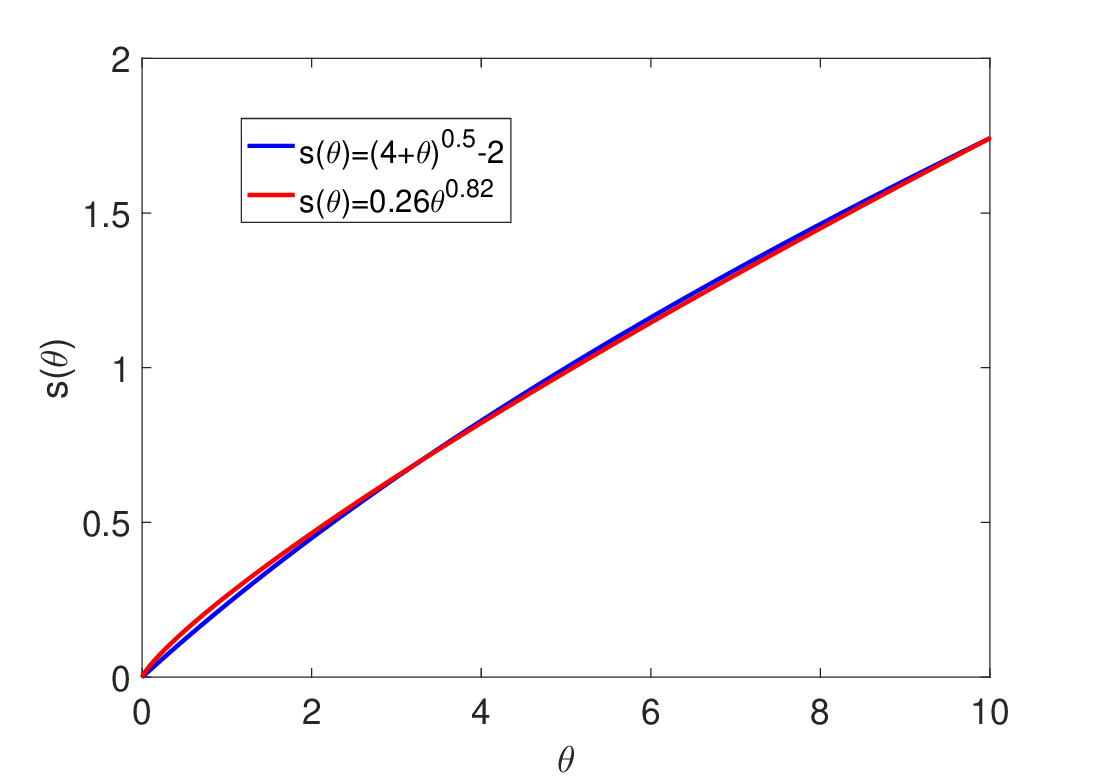}
  \captionof{figure}{Approximation for $s(\theta)=\sqrt{4+\theta}-2$ used by heSRPT.}
  \label{fig:3_4_1}
\end{minipage}
\end{figure}

\section{Discussion: Towards More General Problems}
This paper focuses on the case where all jobs share the same speedup function $s(\theta)$, the system resource budget $B$ is constant over time, and the objective $J$ is a weighted sum of job completion times.  
We presented the CDR Rule and the GWF method, and the SmartFill Algorithm to solve this problem.  
Beyond this specific setting, however, we find that the CDR Rule is applicable to much more general scenarios, including heterogeneous speedup functions across jobs, time-varying resource budgets, and more general system objectives $J$.  

\subsection{General Problem Definition}
Consider a system with $M$ parallelizable jobs, each available at time $t=0$.  
Let $x_i$ denote the size of job $i$ for $i=1,2,\ldots,M$.  
Let $B(t)$ denote the total available resource at time $t$, and let $\theta_i(t)$ denote the allocation to job $i$ at time $t$.  
The allocations must satisfy
\begin{equation}
    \sum_{i=1}^M \theta_i(t) \leq B(t), \quad \forall t > 0.
\end{equation}

Each job $i$ is associated with a time-varying speedup function $s_i(\theta, t)$, which denotes its service rate at time $t$ when allocated bandwidth $\theta$. We assume that for each $i$ and each fixed $t$, the function satisfies: 
\begin{itemize} 
\item $s_i(0, t) = 0$, 
\item $s_i(\theta, t)$ is strictly increasing w.r.t. $\theta$, 
\item $s_i(\theta, t)$ is differentiable w.r.t. $\theta$. 
\item $s_i(\theta, t)$ is strictly concave w.r.t. $\theta$. 
\end{itemize}

Let $s_i'(\theta,t) = \frac{\partial s_i(\theta,t)}{\partial \theta}$.  
We assume that $s_i'(\theta,t)$ is continuous in $\theta$, and that $B(t)$, $s_i(\theta,t)$, $s_i'(\theta,t)$, and $\theta_i(t)$ are all right-continuous in $t$.  

The amount of service received by job $i$ in the interval $[t_1,t_2]$ is
\begin{equation}
    Q_i(t_1,t_2) = \int_{t_1}^{t_2} s_i(\theta_i(t),t)\,dt, \quad \forall t_1<t_2.
\end{equation}
Let $T_i$ denote the completion time of job $i$.  
The system performance metric $J$ is assumed to be a continuous and strictly increasing function of the completion times:
\begin{equation}
    J = f(T_1,T_2,\ldots,T_M).
\end{equation}

The goal is to find the optimal allocation $\theta_i(t)$ for each job $i$ and time $t$ that minimizes $J$, given $B(t)$, $x_i$, $s_i(\theta,t)$, and $f$.  

\subsection{CDR Rule for the General Problem}
Let $\theta^*(t)$ denote an optimal scheduling policy for the general problem.  
We have the following result.  

\begin{theorem}
    (\textbf{CDR Rule for the General Problem})  
    For any pair of jobs $i$ and $j$, there exists a constant $c_{i,j}$ such that
    \begin{equation}
        \frac{s_i'(\theta_i^*(t), t)}{s_j'(\theta_j^*(t), t)} = c_{i,j},
    \end{equation}
    for all $t$ such that $\theta_i^*(t)>0$ and $\theta_j^*(t)>0$.
\end{theorem}

The proof follows directly from the same contradiction argument used in Theorem~\ref{theorem:CDR}.  
This result demonstrates that the CDR Rule is a very general structural property for parallel scheduling and can greatly reduce the search space of optimal schedules.  

\subsection{GWF and SmartFill in the General Problem}
Unlike the CDR Rule, the GWF and SmartFill algorithms cannot be directly extended to the general problem.  
Both rely on knowledge of the optimal job completion order.  
For OPT, the order is given by SJF, but for the general problem the optimal order is unclear and may depend on time-varying system conditions.  
Thus, the general problem remains open.  

We believe there are significant research opportunities in extending the CDR Rule (and potentially GWF) to design new algorithms for these more general parallel scheduling problems.

\section{Conclusion}
\label{sec:conclusion}
In this paper, we solved the open problem of scheduling parallel jobs under general concave speedup functions, a fundamental challenge in modern computing systems.  
We established the CDR Rule, which characterizes the structure of any optimal schedule and dramatically reduces the search space.  
Building on this property, we developed the GWF algorithm to compute allocations under fixed derivative ratios, and combined these insights into the SmartFill algorithm.  
Unlike heSRPT, which allocates resources to all active jobs, SmartFill intelligently determines which jobs should receive resources and how much to allocate.  
We showed that for a broad class of \emph{regular} speedup functions, SmartFill produces closed-form optimal solutions, while for non-regular concave functions it provides numerical solutions.  
Numerical evaluations confirmed that SmartFill consistently outperforms heSRPTacross diverse concave speedup functions.  
Finally, we demonstrated that the CDR Rule generalizes to settings with heterogeneous speedup functions, time-varying resources, and broader system objectives, highlighting its potential as a unifying principle in parallel scheduling.

\end{document}